\documentclass[a4paper,11pt]{article}
\usepackage[margin=1in]{geometry}

\newif\ifFull
\Fullfalse

\usepackage[utf8]{inputenc}
\usepackage{xspace}
\usepackage{paralist}
\usepackage{color}
\usepackage{subfig}
\usepackage{todonotes}
\usepackage{booktabs}
\usepackage{multirow}
\usepackage{hyperref}
\usepackage{amsmath}
\usepackage{amssymb}
\usepackage{amsthm}
\usepackage{pdflscape}

\usepackage[algo2e,ruled,vlined]{algorithm2e}

\newcommand{\presence}{\Psi}
\newcommand{\activity}{\Phi}
\newcommand{\conflict}{C}

\newcommand{\R}{\mathbb{R}}

\newcommand{\GeneralMaxTotal}{\textsc{General\-Max\-Total}\xspace}

\newcommand{\kRestrictedMaxTotal}{$k$-\textsc{Re\-stric\-ted\-Max\-Total}\xspace}

\newcommand{\GreedySearch}{\textsc{Greedy}\xspace}
\newcommand{\LocalSearch}{\textsc{Phased\-Local\-Search}\xspace}
\newcommand{\IntervalSearch}{\textsc{IntGraph}\xspace}
\newcommand{\ILP}{\textsc{Ilp}\xspace}

\newcommand{\LocalSearchShort}{\textsc{PLS}}

\newtheorem{problem}{Problem}
\newtheorem{lemma}{Lemma}

\begin{document}

\title{Temporal Map Labeling:\\ A New Unified Framework with Experiments}
\date{}

\author{
Lukas Barth\thanks{Karlsruhe Institute of Technology, Karlsruhe, Germany}      
\and
Benjamin Niedermann\footnotemark[1]      
\and
Martin N\"ollenburg\thanks{Algorithms and Complexity Group, TU Wien, Vienna, Austria}      
\and
Darren Strash\footnotemark[1]      
}

\maketitle
\begin{abstract}
  The increased availability of interactive maps on the Internet
  and on personal mobile devices has created new challenges in
  computational cartography and, in particular, for label placement in
  maps. Operations like rotation, zoom, and translation dynamically
  change the map over time and make a consistent adaptation of the
  map labeling necessary.

  In this paper, we consider map labeling for the case that a
  map undergoes a sequence of operations over a specified time span. We unify and
  generalize several preceding models for dynamic map labeling into one
  versatile and flexible model. In contrast to previous research, we
  completely abstract from the particular operations (e.g., zoom,
  rotation, etc.) and express the labeling problem as a set of time
  intervals representing the labels' presences, activities, and
  conflicts.  The model's strength is manifested in its simplicity and
  broad range of applications. In particular, it supports label
  selection both for map features with fixed position as well as for
  moving entities (e.g., for tracking vehicles in logistics or air
  traffic control).

  Through extensive experiments on OpenStreetMap data, we evaluate our
  model using algorithms of varying complexity as a case study for
  navigation systems. Our experiments show that even simple (and thus,
  fast) algorithms achieve near-optimal solutions in our model with
  respect to an intuitive objective function.
\end{abstract}

\newcommand{\Change}[1]{\textcolor{blue}{#1}}

\section{Introduction}

Dynamic digital maps are becoming more and more ubiquitous, especially with the rising numbers of location-based services and smartphone users worldwide.
Consumer applications that include personalized and interactive map views range from classic navigation systems to map-based search engines and social networking services.
Likewise, interactive digital maps are a core component of professional geographic information systems.
All these map services have in common that the content of the map view is changing over time based on interaction with the system (i.e., zooming, panning, rotating, content filtering, etc.) or the physical movement of the user or a set of tracked entities.
A key ingredient of every (paper or digital) map are features like geographic places, points of interest, or search results that all need to be labeled by a name or a graphical symbol in order to become meaningful for the map user. While we focus on point features in our experiments, our results can be generalized to other map features, such as line or area features.

For static maps, \emph{labeling}---the selection and placement of labels---is a well-studied research area in cartography and computational geometry. One of the primary quality constraints for labeling, is that no two labels may overlap each other~\cite{fw-ppwalm-91,i-pnm-75}.
Most formalizations of the static map labeling problem are NP-hard and, therefore, a variety of approximation algorithms and heuristics have been proposed in the literature. See~\cite{kb-alppcilp-08} for an overview.

More recently, dynamic map labeling has captured the interest of researchers, leading to the study of labeling problems in maps that support certain subsets of operations like zooming, panning, and rotations.
The main difficulty in dynamic maps is that the selection and placement of labels must be temporally coherent (or \emph{consistent}) during all map animations resulting from interactions, rather than being optimized individually for each map view as in static map labeling.
A map with temporally coherent labeling avoids visually distracting effects like jumping or flickering labels~\cite{Been2006}.
Again, consistent dynamic map labeling problems are typically NP-hard and approximation results as well as heuristics are known~\cite{Been2006,Been2010,gnr-clrm-16,gnr-elsrm-16,Liao2014}.
However, most of the existing algorithmic results in dynamic map labeling take a \emph{global} view on the map, which optimizes over the \emph{whole} interaction space, regardless of which portion of that space is actually explored by the user.

In this paper we take a more local view on dynamic map labeling.
Our aim is to develop algorithms that optimize the labeling for a specific map animation given offline as an input.
Any feature or label that is not relevant for that particular animation---for example, because it never enters the map view---can be ignored by our algorithms.
This approach not only allows us to compute better labelings by removing unnecessary dependencies and non-local effects, but it also reduces the problem size, since fewer features and labels must be taken into account.

We first formulate an abstract, generic framework for offline,
temporal labeling problems, in which labels and potential
conflicts between labels are represented as intervals over time.
To represent label's presence, we use a \emph{presence} interval,
which corresponds to the time that a label is present (but not necessarily displayed) in the
map view. That is, whenever a label enters the map view, a
corresponding presence interval starts, and whenever a label
leaves the view, its current
presence interval ends. Next, a \emph{conflict interval}
(or simply \emph{conflict}) between two present
labels starts and ends at the points in time at which the two labels
start and stop intersecting.  A temporal labeling is then simply
represented as a set of subintervals---the labels' \emph{activity
intervals}, during which the labels are displayed, where no two
conflicting labels are displayed simultaneously.  Depending on the
objective and consistency constraints of the labeling
model, different sets of subintervals may be chosen by the algorithm.

This is a very versatile framework, which includes, for instance, map
labeling for car navigation systems, in which the map view changes
position, angle, and scale according to the car's position, heading,
and speed following a particular route. To give another, seemingly
different example, it also includes the problem of labeling a set of
moving entities in a map view (e.g., for tracking vehicles in
logistics or planes in air traffic control).  Also non-map related
applications such as labeling 3D scenes as they occur in medical
information systems are covered by our model.  Put differently, the
model comprises any application in which start and end times of label
presences and conflicts can be determined in advance.  Further, the
conflicts are not restricted to label-label conflicts but may also
include label-object conflicts.

\paragraph{Our Results}
In a companion paper~\cite{Gemsa2013} we investigated the
  underlying models from a theoretical point of view, showing NP-hardness and W[1]-hardness
  for optimization problems in these models. We further
  provided optimal integer linear programming approaches and 
  approximation algorithms, 
  but without any experimental results.
  
  In this paper, we build upon our previous work, present more sophisticated heuristic algorithms, 
  and provide an extensive experimental evaluation of our
  proposed temporal labeling models and algorithms in a case study for navigation systems.
  Our experiments illustrate the usefulness of our models for this
  application, and further show the strength of each algorithm
  under each model. Ultimately, our experiments show that that simple but fast
  algorithms achieve near-optimal solutions for the optimization problems---which is very encouraging, given the hardness results.
  Lastly, while the models in~\cite{Gemsa2013} were developed specifically for
  navigation systems, we adapt the models to make them more broadly applicable to temporal map labeling scenarios.

\subsection{Related Work}
We now systematically review prior
research on label placement in maps focusing on dynamic map labeling.

In 2003, Petzold et al.~\cite{Petzold2003}
presented a framework for automatically placing labels on dynamic
maps. They split the label placement procedure into two phases, namely
a (possibly time-consuming) pre-processing phase and a query phase which
computes the labeling of custom-scale maps.  However, this approach does not
guarantee that labels do not \emph{jump} or \emph{flicker} while
transforming the map.

In 2006, Been et al.~\cite{Been2006} introduced the first formal
model for dynamic maps and dynamic labels, formulating a general
optimization problem.
They described the change of a map by
the operations \emph{zooming}, \emph{panning}, and \emph{rotation}. In order to avoid \emph{flickering} and \emph{jumping} labels
while transforming the map with zooming and panning, they required four desiderata for
\emph{consistent} dynamic map labeling. 
These comprise \emph{monotonicity}, that labels should not vanish when zooming in or appear when zooming out (or any of the two when panning), \emph{invariant point placement}, where label positions and size remain invariant during movement, and \emph{history independence}---placement and selection of labels should be a function of the current map state only. Monotonicity was modeled as selecting for each label at most one scale interval, the so-called \emph{active range}, during which the label is displayed.
They introduced the
\emph{active range optimization problem} (ARO) maximizing the sum of
active ranges over all labels such that no two labels overlap and all desiderata are fulfilled. They proved that ARO is NP-hard for
star-shaped labels and presented an optimal greedy algorithm for a
simplified variant. 

That model was the point of departure for several subsequent papers
considering the operations \emph{zooming}, \emph{panning} and
 \emph{rotation},  mostly independently. Been et al.~\cite{Been2010} took
a closer look at different variants of ARO for zooming. They
showed NP-hardness 
and gave approximation
algorithms. In the same manner further variants were investigated by
Liao et al.~\cite{Liao2014}. Gemsa et al.~\cite{Gemsa2011}
presented a fully polynomial-time approximation scheme (FPTAS) for a
special case of ARO, where the given map is one-dimensional and only
zooming is allowed. However, they combined the selection problem with a placement problem in a slider model.
Zhang et al.~\cite{Zhang2015} also considered the model of Been et
al.~\cite{Been2006} for zooming, however, instead of maximizing the
total sum of active ranges, they maximized the minimum active
range among all labels. They discussed similar variants as Liao et
al.~\cite{Liao2014} and Been et al.~\cite{Been2010}, also proving
NP-hardness and giving approximation algorithms.

Gemsa et al.~\cite{gnr-clrm-16,gnr-elsrm-16} extended the ARO
model to \emph{rotation} operations. They first showed
that the ARO problem is  NP-hard in that setting and introduced an
efficient polynomial-time-approximation scheme (EPTAS) for unit-height
rectangles~\cite{gnr-clrm-16}. In a second step they experimentally
evaluated heuristics, algorithms with approximation guarantees, and
optimal approaches based on integer linear
programming~\cite{gnr-elsrm-16}.  A similar setting for rotating maps was
considered by Yokosuka and Imai~\cite{Yokosuka2013}. Instead
of ARO, they aimed at finding the maximum font
size for which all labels can always be displayed without
overlapping.

Apart from the results based on the consistency model of Been et
al.~\cite{Been2006}, other approaches have been considered, too. Maass
et al.~\cite{Maass2006} described a view management system for
interactive three-dimensional maps of cities also considering label
placement. Mote~\cite{Mote2007} presented a fast label placement
strategy without a pre-processing phase.
Luboschik~\cite{Luboschik2008} described a fast particle-based strategy
that locally optimizes the label placement. All these approaches have in
common that they do not take the consistency criteria for dynamic map labeling into account.

A different generalization of static point labeling is dynamic point
labeling. In this case not the map is being transformed, but the point set
changes by adding or removing points as well as by moving points
continuously.  Inspired by air-traffic control, De Berg and
Gerrits~\cite{DeBerg2013} considered moving points on a static map that
all must be labeled. They presented a sophisticated heuristic for
finding a reasonable trade-off between label speed and label
overlap. Finally, Buchin and Gerrits~\cite{Buchin2014} showed that
dynamic point labeling is strongly PSPACE-complete.

\section{Model}
\begin{figure}
  \centering
 \includegraphics[page=1]{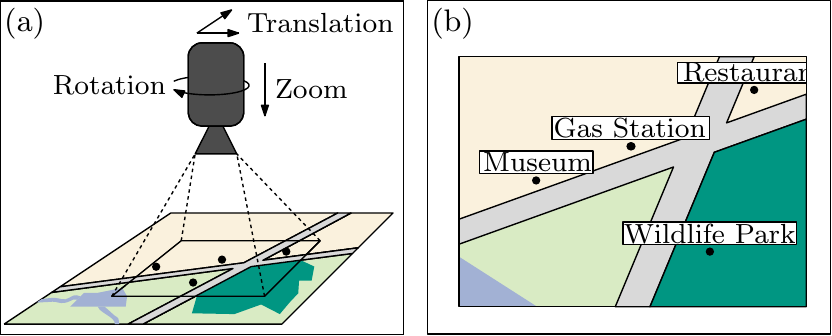}
 \caption{\small Model.~(a) Overall view of the scene.
   Depending on the rotation, translation, and zoom, the camera shoots
   a restricted part of the scene. The objects to be labeled are
   represented by black dots.~(b) The corresponding viewport of the camera. The labels are placed near their objects.}
 \label{fig:model:scene}
\end{figure}

We now formally describe the temporal labeling model that we use
in the remainder of the paper.
It unifies and generalizes the models
presented by Been et al.~\cite{Been2006}, De Berg and
Gerrits~\cite{DeBerg2013} as well as the model that we
introduced in~\cite{Gemsa2013}. The notation is mainly adopted
from~\cite{Gemsa2013}.

\subsection{Basic Model}
We are given a set $O=\{o_1,\dots,o_n\}$ of objects in a scene over a
given time span $\mathcal T= [0,T]$.  Further, for each object~$o$ we
are given a label~$\ell$, e.g., text describing $o$. We denote the
set of labels by $L=\{\ell_1,\dots,\ell_n\}$, where $\ell_i$ is the
label of $o_i$. To quantify the importance of a label, we define for
each label~$\ell \in L$ a positive weight $w_\ell\in \mathbb R^+$.

We have a restricted view on the scene through a camera, i.e., the
objects are projected onto an infinite plane $P$ such that we can only
see a restricted section~$V$ of~$P$, where $V$ models the
\emph{viewport} of the camera; see Fig.~\ref{fig:model:scene} for an
example.
During the time interval~$\mathcal T$, the
objects are moving and the camera changes its perspective by changing
its position, direction and zoom. We denote the plane $P$ and the
viewport $V$ at time $t$ by $P(t)$ and $V(t)$, respectively. Depending
on the position of the object $o_i\in O$, each label $\ell_i$ has a
certain shape and position on $P(t)$ at time $t$; we denote the
geometric shape of $\ell$ at time $t$ by $\ell(t)$. Following
typical map labeling models we may assume that $\ell(t)$ is a (closed) rectangle
enclosing the text; one may also assume other shapes.  In the
following we introduce some further notations to describe the setting precisely.

According to the perspective and position of the
camera, not every label $\ell(t)$ is contained in the viewport at
time~$t$. We say that a label $\ell$ is \emph{present} at time~$t$ if
$\ell(t)$ is (partly) contained in $V(t)$; that is, $\ell(t)\cap V(t)\neq
\emptyset$. We assume that the time intervals, during which a label~$\ell$ is
present, are given by a set $\presence_\ell$ of disjoint, closed
sub-intervals of $\mathcal T$; see Fig.~\ref{fig:justified} and Fig.~\ref{fig:activity}.
For such
an interval $[a,b] \in \presence_\ell$ we also write $[a,b]_\ell$
indicating that it belongs to~$\ell$. We denote the union of all those
sets $\presence_\ell$ by $\presence$ and assume that $\presence$ is a
multi-set, as it may contain the same interval $[a,b]$ multiple
times, where each occurrence of $[a,b]$ belongs to a different label.

Two labels $\ell$ and $\ell'$ are in \emph{conflict} at time $t\in
\mathcal T$, if the geometric shapes of both labels intersect, i.e.,
$\ell(t)\cap \ell'(t) \neq \emptyset$.
Following~\cite{Gemsa2013} we describe the occurrences of conflicts
between two labels~$\ell,\ell'\in L$ by a set of closed
intervals:~$\conflict_{\ell,\ell'}=\{[a,b]\subseteq \mathcal T\mid
[a,b]$ is maximal and~$\ell$ and $\ell'$ are in conflict at all
$t\in[a,b]\}$. For such an interval $[a,b]\in\conflict_{\ell,\ell'}$
we also write $[a,b]_{\ell,\ell'}$ indicating that it is a
\emph{conflict interval} between $\ell$ and $\ell'$. We denote the set
of all conflict intervals over all pairs of labels by the multi-set
$\conflict$.

To avoid overlaps between labels, we display a label~$\ell$ only at
certain times when no other displayed label overlaps $\ell$; the label
$\ell$ is said to be \emph{active} at those times. We describe the activity
of~$\ell$, by a set~$\activity_\ell$ of disjoint
intervals\footnote{Technically, one needs to distinguish between open and closed intervals, i.e., for closed rectangular labels, the presence and conflict intervals are closed but the activity intervals are open. However, including or excluding the interval boundaries makes no difference in our algorithms and hence we decided to simply use the notation $[a,b]$ for all respective intervals unless stated otherwise.}. For such an interval $[a,b]\in \activity_\ell$ we also
write $[a,b]_\ell$ to indicate that the \emph{activity interval}
belongs to $\ell$. The union of all activity intervals over all labels
is denoted by the multi-set $\activity$.

We say that two activity intervals $[a,b]_\ell$ and $[c,d]_{\ell'}$ of
two labels $\ell$ and $\ell'$ are in \emph{conflict} if there is a
time $t$ in the intersection of the open intervals $(a,b) \cap (c,d)$ such that the labels $\ell$ and $\ell'$ are in conflict at $t$.

An instance of temporal labeling is then defined by the set $L$ of
labels, the set $\presence$ of presence intervals and the set $\conflict$ of
conflict intervals. We thus completely abstract away the geometry of the problem, while all essential information of the temporal labeling instance is captured combinatorially in~$\presence$ and~$\conflict$.
In this paper, we primarily focus on conflict-free label selection, and therefore assume that~$\presence$ and $\conflict$ are given as input. However, in Section~\ref{sec:workflow} we describe how to construct~$\presence$ and~$\conflict$ for the specific application of navigation systems.

\begin{figure}[t]
 \centering
 \includegraphics[page=5]{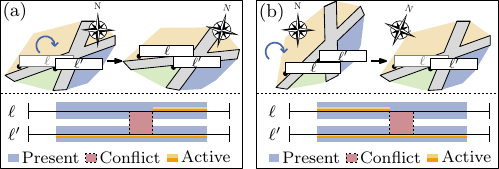}
 \caption{\small Label activity. The maps rotate clockwise, while the labels keep aligned horizontally. Black labels are active, while gray labels are
   inactive. The intervals illustrate presence, conflict and activity
   intervals.  The witness label $\ell'$ justifies (a) the start (b)
   the end of $\ell$'s activity interval.  }
 \label{fig:justified}
\end{figure}

Similarly to Been et al.~\cite{Been2006} for a temporal labeling we require the following temporal
consistency criteria:
\begin{compactenum}[(C1)]
\item A label should not be set active and inactive repeatedly to avoid \emph{flickering}.
\item The position and size of a label should be changed continuously, it should not \emph{jump}.
\item Labels should not overlap.
\end{compactenum}

We formalize those consistency criteria and say the the activity set $\activity$ is \emph{valid} (see Fig.~\ref{fig:activity}) if
\begin{compactenum}[(R1)]
\item for each activity interval $I_\ell \in \activity$ there is a presence interval $I'_\ell \in \presence$ with $I_\ell \subseteq I'_\ell$,\label{req:active-impl-present}
\item for each presence interval~$I_\ell \in \presence$ there is at most one activity interval $I'_\ell \in \activity$ with $I'_\ell\subseteq I_\ell$, \label{req:one-active-interval}and
\item no two activity intervals of $\activity$ are in conflict.  \label{req:overlap-free}
\end{compactenum}
    Requirement (R1) enforces that a label is only
displayed if it is present in the viewport. Requirement (R2)
  prevents a label from \emph{flickering} during a presence interval (C1),
  while (R3) enforces that no two displayed labels overlap (C3). In fact,
  (R2) is only a minimum requirement for avoiding flickering labels,
  which we later extend to stronger variants. By assuming that
labels' positions are fixed relative to their anchors, labels may not \emph{jump} (C2) as long as labeled objects are either fixed or move continuously.
From now on we assume that an activity set is valid, unless we
state otherwise.

\subsection{Optimization Problem}

Based on the introduced model we investigate two optimization problems
for temporal labeling that aim to maximize the overall active time
of labels. The first problem allows for any number of
labels to be active at the same time, and the second allows at most
$k$ labels to be active at the same time, which reduces the amount of
presented information.
We define the weight of an activity interval $[a, b]_\ell \in \activity$ to be $w([a, b]_\ell) = (b-a) \cdot w_\ell$.

\begin{problem}[\textsc{GeneralMaxTotal}]\

\begin{tabular}{ll}
\emph{\textbf{Given:}}& Instance~$(L,\presence,\conflict)$.  \\
\emph{\textbf{Find:}}& Activity set $\activity$ maximizing~
$\sum_{[a,b]_\ell\in \activity}w([a,b]_\ell)$.\\
\end{tabular}
\end{problem}

Figure~\ref{fig:example-labeling:gmt} shows an example of a single frame of a temporal labeling that is
  optimal with respect to \GeneralMaxTotal. While such a labeling is
  acceptable for general applications such as spatial data exploration,
  for small-screen devices, such as car navigation systems, the same labeling
  may overwhelm or distract the user with too much
  additional information. In fact, psychological studies have shown
  that untrained users are strongly limited in receiving, processing,
  and remembering information (e.g., see \cite{Miller1956}). For
  applications that do not receive a user's full attention it is therefore desirable to restrict the number of
  simultaneously displayed labels, which we formalize as an
  alternative optimization problem as follows.

\begin{problem}[$k$-\textsc{RestrictedMaxTotal}]\

\begin{tabular}{ll}
\hspace{-2ex}\emph{\textbf{Given:}}& Instance~$(L,\presence,\conflict)$, $k\in \mathbb N$.  \\
\hspace{-2ex}\emph{\textbf{Find:}}& Activity set $\activity$ maximizing~
$\sum_{[a,b]_\ell\in \activity}w([a,b]_\ell)$,
 s.t.\ at any time $t$ at most $k$ labels\\ & are active.
\end{tabular}
\end{problem}

In~\cite{Gemsa2013} we showed that \GeneralMaxTotal is NP-hard
and W[1]-hard. By W[1]-hardness, we cannot expect
algorithms for \kRestrictedMaxTotal that are fixed-pa\-ra\-me\-ter
tractable on $k$. This even applies for the example of navigation systems that we consider in Sect.~\ref{sec:transformation}. We therefore focus on heuristics for the two problems.

\begin{figure*}[t]
 \centering
 \includegraphics[page=8]{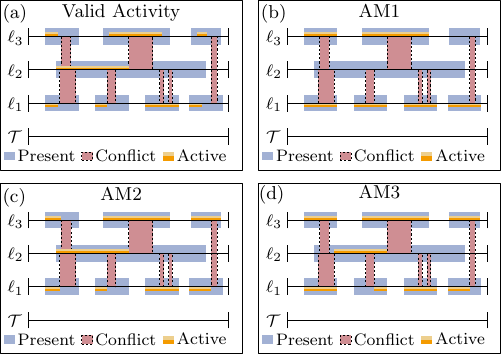}
 \caption{\small Valid activities and the activity models AM1, AM2 and AM3.}
 \label{fig:activity}
\end{figure*}
\subsection{Activity Models}

So far labels may become active or inactive within the viewport
without any external influence, see, e.g., the second activity interval of $\ell_3$ in Fig.~\ref{fig:activity}(a). Hence, the activity behavior of labels, even in an optimal solution $\activity$, is not
necessarily explainable to a user by simple and direct observations such as ``the label
becomes inactive at time~$t$, because at that moment an overlap starts
with another active label''. The absence of those simple logical
explanations may lead to unnecessary irritations of the user.  To
account for that we introduce the concept of justified activity intervals.

Consider a label $\ell$ with activity interval $[a,b]_\ell\in
\activity$. We say that the start of $[a,b]_\ell$ is
\emph{justified} if $\ell$ enters the viewport at time $a$ or if there is
a \emph{witness} label $\ell'$ such that a conflict of $\ell$ and
$\ell'$ ends at~$a$ and $\ell'$ is active at~$a$; see Fig.~\ref{fig:justified}(a).

Analogously, we say that the end of $[a,b]_\ell$ is \emph{justified}
if $\ell$ leaves the viewport at time $b$ or if there is a witness label
$\ell'$ such that a conflict of $\ell$ and $\ell'$ begins at $b$ and
$\ell'$ is active at $b$; see Fig.~\ref{fig:justified}(b). If both the
start and end of $[a,b]_\ell$ are justified, then $[a,b]_\ell$ is
justified.

Following our preceding paper~\cite{Gemsa2013}, we distinguish the
three activity models AM1, AM2, and AM3 that consider justified
activity intervals; see Fig.~\ref{fig:activity}.  While for AM1, a label may only become active and
inactive when it enters and leaves the viewport, for AM2 it may also
become inactive before leaving the viewport if a witness label justifies this event. AM3 further
allows a label to become active after entering the viewport if a witness label justifies that event.

\textbf{AM1.} An activity $\activity$ satisfies AM1 if
any activity interval $[a,b]_\ell \in \activity$ is justified and there is a presence
interval $[c,d]_\ell \in \presence$ of the same label $\ell$ with $[a,b]_\ell =[c,d]_\ell$.

\textbf{AM2.}  An activity $\activity$ satisfies AM2 if
any activity interval $[a,b]_\ell \in \activity$ is justified and there is a presence
interval $[c,d]_\ell \in \presence$ of the same label $\ell$ with $a=c$.

\textbf{AM3.}  An activity $\activity$ satisfies AM3 if
any activity interval $[a,b]_\ell \in \activity$ is justified.

We have described only the core of the model. Depending on the application it can be
easily extended to more complex variants, e.g., requiring minimum
 activity times.

\section{Workflow}\label{sec:workflow}
In this section we describe a simple but flexible workflow for temporal
labeling problems. This workflow consists of two phases. In the first
phase a concrete geometric labeling problem is transformed into an abstract
temporal labeling instance $I=(L,\presence,\conflict)$. This step critically depends on the concrete geometric model of the
given temporal labeling problem. Here, we consider the application
of a car navigation system; other labeling problems, such as labeling moving
entities, can be handled similarly. In the second phase, either
\GeneralMaxTotal or \kRestrictedMaxTotal is solved for the output instance $I$ from the first phase. 
We now describe these two phases in greater detail.

\subsection{Phase 1 -- Transformation into Intervals}\label{sec:transformation}
This phase depends on the specific labeling problem given. It
transforms the input for a particular geometric setting into a temporal labeling instance that can then be handled independently from the geometry.

\begin{figure}[t]
    \centering
    \includegraphics[page=6]{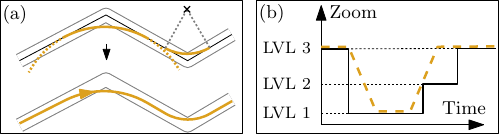}
    \caption{\small Trajectory.~(a) The corners of the selected route (black
      polyline) are smoothed by circular arcs obtaining a continuous
      and differentiable trajectory (orange curve). (b) The black
      rectilinear curve shows the zoom levels assigned to the
      underlying roads over time. The orange dashed line illustrates
      the interpolated actual zooming of the viewport. }
    \label{fig:trajectory}
 
\end{figure}

\paragraph{Example: Navigation Systems} For our experiments
we consider the use case of car navigation systems. In this use case
the viewport of the map moves along a selected route to the
journey's destination such that the camera is perpendicular to the
map; that is, the user of the navigation system observes the map in aerial perspective such that at any time the viewport has a certain position, rotation, and scale. See Fig.~\ref{fig:example_trajectory}.

As in~\cite{Gemsa2013}, we model the
viewport as an arbitrarily oriented rectangle~$R$ that
defines the currently visible part of the map on the screen.  The
viewport follows a trajectory that we model as a continuous
differentiable function $\tau\colon \mathcal T \to \R^2$.

In our setting, we obtain $\tau$ from a polyline describing the
selected route by \emph{smoothing} the polyline's corners by circular
arcs; see Fig.~\ref{fig:trajectory}(a). Thus, $\tau$ is described by a
sequence of line segments and circular arcs.

The viewport is described by a function~$V\colon
\mathcal T\to \mathbb{R}^2\times [0,2\pi]\times[0,1]$. The interpretation of
$V(t)=(c(t),\alpha(t),z(t))$ is that at time~$t$ the center of~$R$ is located
at~$c(t)$,~$R$ is rotated clockwise by the angle~$\alpha(t)$ relatively to a
north base line of the map, and $R$ is scaled by the factor $z(t)$. We
call~$z(t)$ the \emph{zoom} of $V$ at time~$t$.
Since~$R$ moves along~$\tau$ we define $c(t) = \tau(t)$.
To avoid
distracting changes of the map, we assume that the viewport does not
both rotate and zoom at the same time. More precisely, we are given a
finite set~$\mathcal Z$ of zoom levels at which the viewport is
allowed to rotate. Hence, when the camera zooms, the
trajectory must form a straight line for that particular period of time.

The objects in $O$ describe points of interests and are fixed on the
map. We model a label $\ell$ of an object $o\in O$ as a rectangle on
the plane~$P$ that is anchored at the projection of~$o$ onto $P$ with
the midpoint of its bottom side. It does not change its size on the
screen over time.  To ensure good readability, the
labels are always aligned with the viewport axes as the viewport
changes its orientation (i.e., they rotate around their anchors by the
same angle~$\alpha(t)$); see Fig.~\ref{fig:example_trajectory}.

For each label we compute the time events when it enters or leaves the
viewport, and when it starts and stops overlapping another label. Since
rotation and zooming are temporally separated, those operations can be
considered independently.  Computing the time events for
rotations requires an intricate geometric analysis, which is described
in~\cite{n-cldmust-12}. For changing from one zoom level to another, we
do not allow instantaneous changing of zoom levels, but instead we
linearly interpolate the scale of the map between both zoom levels,
as in Fig.~\ref{fig:trajectory}(b). (In our experiments we further enforce a
minimum duration between two changes of zoom levels to avoid
oscillation effects; see Sect.~\ref{sec:evaluation}.)
Under these conditions, time events for zooming can be computed by detecting
collisions among linearly moving objects.

The computed time events directly translate into presence and conflict
intervals of the labels. Hence, 
we obtain the  temporal labeling instance $I=(L,\presence,\conflict)$.

\begin{figure}[t]
  \centering
 \centering
 \includegraphics[page=3]{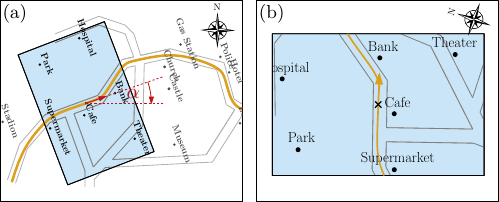}
 \caption{\small Trajectory-based labeling.~(a) Viewport moves and aligns
   along a given trajectory (fat orange line). Labels align to the
   viewport.~(b) The user's view on the scene. }
 \label{fig:example_trajectory}
\end{figure}

\paragraph{Other Scenarios} Our model is not restricted to labels
of point features, but it also can be applied to labels of other features
such as line and area features. For example, one could pre-compute a
label placement for roads and combine the road labels with labels for point features by computing all temporal conflict events. Thus, we again
obtain a temporal labeling instance~$I=(L,\presence,\conflict)$
describing the setting. By pre-selecting active intervals for certain
labels, we can further enforce that they are definitely active at
the selected times. In the same manner we can ensure that labels do not
overlap certain important map features. Finally, we do not require the labeled
objects to be fixed, but they may also move. As long as the start and
end times of label presence- and conflict intervals can be determined in
advance, they can be represented in our model. Depending on the
 setting, this may involve non-trivial geometrical
computations, but once the transformation is done, the different
scenarios are treated~equally.

\subsection{Phase 2 -- Resolving Conflicts}\label{sec:phase2}
In the second phase we compute the activity intervals for all labels.  We present optimal approaches as well as efficient heuristics for solving \GeneralMaxTotal and \kRestrictedMaxTotal on
$I=(L,\presence,\conflict)$.

\subsubsection{Integer Linear Programming}
In order to provide upper bounds for the evaluation of our
labeling algorithms, we implement an integer linear programming (ILP)
model that solves \GeneralMaxTotal and \kRestrictedMaxTotal
optimally. We introduced this ILP formulation in~\cite{Gemsa2013}
and refrain from repeating it here. The key idea is to split the
given intervals into disjoint \emph{elementary intervals}. Those
intervals are then optimally combined for the solution such that the
specific constraints of \GeneralMaxTotal and \kRestrictedMaxTotal are
satisfied with respect to the chosen activity model.
Finding an optimal solution for an ILP model is NP-hard in
general. However, it turns out that in practice we can apply
specialized solvers to find optimal solutions for reasonably sized
instances in acceptable time; see Sect.~\ref{sec:evaluation} for details. Hence,
this ILP-based method provides a simple and generic way to produce
optimal solutions. We call this approach~\ILP.

\subsubsection{Approaches Based on Conflict Graphs}
We reduce \GeneralMaxTotal to an independent set problem on a weighted conflict
graph $G=(V,E)$ such that the maximum weight independent set in $G$
induces the optimal solution of~$I$. Since AM3 is the most general
model we first describe the reduction for this variant and the sketch
adaptations for AM1 and AM2.

Let $[a,b]_\ell \in \presence$ be a presence interval of the
label~$\ell\in L$. If~$\ell$ becomes active within $[a,b]_\ell$, then
this happens either at time $a$ or at the end of one of the conflict
intervals of $[a,b]_\ell$. Let $s_1,\dots,s_h$ denote those times.
Analogously, if $\ell$ becomes inactive in $[a,b]_\ell$, then
this happens either at time $b$ or at the beginning of one of the
conflict intervals of $[a,b]_\ell$. Let $t_1,\dots,t_h$ denote those
times.

Hence, if~$\ell$ is active for an interval $[s,t]_\ell\subseteq[a,b]_\ell$,
then there are $s_i$ and $t_j$ with $s_i\leq t_j$ such that
$[s,t]_\ell=[s_i,t_j]_\ell$. We call $[s_i,t_j]_\ell$ a \emph{candidate}.

We construct the graph $G_\text{AM3}=(V,E)$ as follows. For any
presence interval $[a,b]_\ell$ of any label~$\ell$ we introduce a
vertex for any candidate $[s_i,t_j]_\ell$ of $[a,b]_\ell$; we identify
the vertices with their candidates and assign to each vertex the
weight of the candidate. For two candidates $u$ and $v$ of the same
presence interval we introduce the edge $\{u,v\}$. Thus, the
candidates of the same presence interval form a clique~$C$ in~$G$,
which we call a \emph{cluster}. For two candidates of different
presence intervals $[a,b]_\ell$ and $[c,d]_{\ell'}$ we introduce an
edge if and only if $[a,b]_\ell$ and $[c,d]_{\ell'}$ are in conflict
during the intersection of both candidates; we say that the corresponding
candidates are in conflict.

Conceptually, to construct $G_\text{AM2}$, we remove each candidate
from $G_{\text{AM3}}$ that does not start at the beginning of its
presence interval. Further removing each candidate that does not end at the
end of its presence interval gives use graph~$G_\text{AM1}$. Note, however,
that in our implementation we constructed $G_\text{AM1}$ and $G_\text{AM2}$
directly without $G_\text{AM3}$.

Then an independent set~$\mathcal I$ of $G_{\text{AM3}}$ is precisely
a set of candidates that are not in conflict.
We interpret $\mathcal I$ as an activity set of the given instance.
We call $\mathcal I$  \emph{saturated}, if there are no two candidates
$v\in \mathcal I$ and $v'\in V\setminus \mathcal I$ such that
$\mathcal I'=\mathcal I\cup\{v'\}\setminus\{v\}$ is an independent
set, $v'$ and $v$ belong to the same cluster and $w(\mathcal
I) < w(\mathcal I')$, where $w(\mathcal I)=\sum_{u\in \mathcal I}w(u)$. Note that any maximum weight
independent set of $G$ is also saturated.

\begin{lemma}
  Let $\mathcal I$ be a saturated independent set of~$G_{\text{AM}X}$, then
  $\mathcal I$ is a valid activity set of the instance $I$ with respect to AM$X$ where $X\in \{1,2,3\}$.
\end{lemma}

\begin{proof}
  We  prove the lemma only for AM3; similar arguments apply for the
  other two models.  Consider the labeling that we obtain by setting
  the labels' activities according to $\mathcal I$. By
  construction of the candidates, each activity interval in
  $\mathcal I$ is contained in a corresponding presence
  interval (R1). By construction of the clusters each label is set active at
  most once for each presence interval (R2). Further, no two
  labels overlap, because candidates in conflict mutually exclude each
  other in any independent set of~$G_{\text{AM}3}$ (R3).

  We now prove that $\mathcal I$ satisfies AM3 by contradiction. We consider two cases. In the first case there
  is a label $\ell$ that is active during a presence interval
  $[a,b]_\ell$ such that $\ell$ becomes active at time $s$ with $a<s$
  and there is no witness label $\ell'$ such that a common conflict
  ends at $s$. By construction there is an interval $[s,t]_\ell$ in $\mathcal I$ for some $t$.
  Since $a<s$ there is a further candidate $[s',t]_\ell$
  with $s'<s$. Further, we can choose $s'$ such that $[s',t]_\ell$ is
  not in conflict with any candidate of $\mathcal
  I\setminus\{[s,t]_\ell\}$. Hence, $\mathcal I'=\mathcal
  I\cup\{[s',t]_\ell\}\setminus\{[s,t]_\ell\}$ is an independent set
  of $G_{\text{AM}3}$ such that $w(\mathcal I') > w(\mathcal
  I)$. Consequently, $\mathcal I$ is not a saturated independent set,
  which contradicts the assumption.

  In the second case there is a label $\ell$ that is active during a
  presence interval $[a,b]_\ell$ such that $\ell$ becomes inactive at
  time $t<b$ and there is no witness label $\ell'$ such that a common
  conflict begins at $t$. Analogous to the first case, we can show that this implies that $\mathcal I$ is not saturated.
\end{proof}

We use different general heuristics for computing independent sets on
$G_{\text{AM}X}$ for $X\in\{1,2,3\}$. However, those independent sets
are not necessarily saturated so that they do not necessarily 
satisfy the according activity model.  Thus, in a post-processing step, we
check whether the activity $\mathcal I$
satisfies AM$X$. If this is not the case, then there is a
cluster with two vertices $v\in \mathcal I$ and $v'\not\in \mathcal I$
such that $\mathcal I'=\mathcal I\cup\{v'\}\setminus\{v\}$ is an
independent set, and $w(\mathcal I) < w(\mathcal
I')$. We exchange $v$ with~$v'$ and repeat the procedure
until $\mathcal I$ is saturated.

We use the following heuristics for computing an independent set
$\mathcal I$ on $G_{\text{AM}X}$.

\GreedySearch. We first consider \GeneralMaxTotal. Starting with an empty solution~$\mathcal I$, the
algorithm removes the candidate~$c$ with largest weight from $G_{\text{AM}X}$ and adds it to $\mathcal I$. Then, it removes
all candidates from $G_{\text{AM}X}$ that are in conflict with $c$. We
repeat this procedure until all candidates are removed from~$G_{\text{AM}X}$. Since we always take the candidate with largest weight, $\mathcal I$ is saturated.

\begin{figure}[tb]
  \centering
  \includegraphics[page=7]{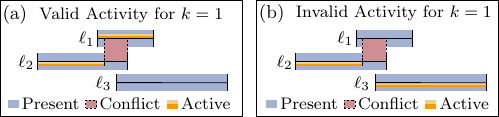}
  \caption{\small Activity of labels for \kRestrictedMaxTotal in AM2 and AM3 for $k=1$. (a) The optimal solution. (b) A solution wrongly produced on $G_\text{AM2}$, $G_\text{AM3}$, respectively. For example, \GreedySearch first adds the presence interval of $\ell_3$ to the solution $\mathcal I$. Then it adds the prefix of $\ell_2$'s presence interval~$P$ that ends at the beginning of the conflict with $\ell_1$. It cannot add the whole presence of $\ell_2$, because otherwise more than one label is active at the same time. For the same reason it cannot add any part of $\ell_1$'s presence interval to the solution $\mathcal I$. Hence, the end of~$P$ is not justified (only $\ell_1$ could justify that end). Hence, $\mathcal I$ is not valid. }
\label{apx:fig:invalid-example}
\end{figure}

In order to solve~\kRestrictedMaxTotal for AM1, we create the
graph $G_{\text{AM1}}$ and apply the procedure as described
above. However, this time we remove not only all candidates that are
in conflict with the candidate $c$, but also any candidate that cannot
be added to~$\mathcal I$ without violating the requirement that at
most $k$ labels are active at the same time. The resulting activity
set $\mathcal I$ is then valid with respect to AM1.  For AM2
  and AM3 we cannot apply the same procedure on $G_\text{AM2}$ and
  $G_\text{AM3}$, respectively, without potentially violating the
  requirement of label witnesses; see also Fig.~\ref{apx:fig:invalid-example}. In our evaluation we therefore use
  the solutions of AM1 instead, which trivially satisfy AM2 and AM3.

\LocalSearch.
\ifFull
\else
As a further method to find high-quality solutions for~\GeneralMaxTotal, we investigated local search algorithms for finding a large-weight independent set in the conflict graph $G_{\text{AM}X}$. As far as we are aware, the Phased Local Search algorithm (\LocalSearchShort) by Pullan~\cite{Pullan2006}, originally developed for the maximum (unweighted) clique problem, is the only local search algorithm that has been shown to find maximum or near-maximum independent sets on weighted versions of standard benchmark graphs~\cite{Pullan2009214}. Other local search algorithms~\cite{Jin201520} may give higher quality solutions for the unweighted case, but they apply operations that serve only to expand the cardinality of the independent set, which may decrease its weight during the process.

An iteration of PLS consists of repeated \emph{improvements}, which add a vertex to a current independent set $\mathcal I$ until it is maximal, followed by a \emph{plateau search}, which swaps a vertex in $\mathcal I$ for one that has one neighbor in $\mathcal I$. When no improvement or swap can be made, $\mathcal I$ is perturbed to include a random vertex. To ensure sufficient diversity of solutions, vertices which are in $\mathcal I$ at the end of an iteration are \emph{penalized}, making them less likely to be considered in future iterations. Vertices recover from penalties by a \emph{penalty decrease} mechanism, where penalties are reduced according to a dynamically updated penalty delay parameter. See~\cite{Pullan2006} for further~details.

PLS proceeds in three phases, each of which performs iterations using one of three specified vertex selection criteria for choosing an improvement/swap among available candidates, uniformly at random: (1) a \emph{random selection} phase, which selects from all available candidates; (2) a \emph{penalty selection} phase, which selects from candidates with the lowest penalty; and (3) a \emph{greedy selection} phase, which selects from candidates with the lowest degree. The standard PLS algorithm performs $50$ iterations of greedy selection, followed by $100$ iterations of penalty selection, and 50 iterations of greedy selection, until a stopping criteria is met.
\fi

\subsubsection{Approach Based on Interval Graphs}
The set of presence intervals $\presence$ induces an \emph{interval
  graph}~$H$. In this graph the presence intervals form the vertex set and two vertices are connected by an edge if and only if the
corresponding intervals intersect. We identify the vertices with the
intervals. In particular each vertex has the weight of its presence
interval. The next approach makes use of $H$ to compute the activity set $\activity$.

\IntervalSearch. We first consider \GeneralMaxTotal and repeatedly apply the following procedure on~$H$
until all vertices are removed from $H$.  We compute a maximum-weight
independent set~$\mathcal I$ on $H$, which can be done in linear
time for interval graphs~\cite{Hsiao1992}. We remove those vertices from $H$ and add the
intervals to the solution~$\activity$. In case of AM1, we remove also
any neighbor of those vertices from $H$. For AM2, we do
not remove those neighbors, but rather shorten the
according presence intervals to the longest prefixes that are not in
conflict with any presence interval of $\mathcal I$.
 For AM3, we
shorten any presence interval of the neighbors to the longest prefix,
infix or suffix that is not in conflict with any presence interval of
$\mathcal I$.
Vertices with empty intervals are removed.
By design, the activity set $\activity$ is valid according to the
applied activity model.

When solving \kRestrictedMaxTotal, 
we abort the procedure after the $k$-th iteration. Since each iteration
computes a set of pairwise disjoint intervals, in the computed activity set $\activity$ at most $k$ labels are active at any time.

\section{Experiments}\label{sec:evaluation}

In this section we present the experimental evaluation of the
different models and algorithms for temporal map labeling considering
the application of navigation systems\footnote{All source code and data instances are freely available at~\url{http://i11www.iti.kit.edu/temporallabeling/}.}. To that end we computed a set
of 204 trajectories on the city map of Berlin, which are between $1$km
and $49$km long, with an average length of $20$km. We measure the
\emph{complexity} of the instances by their input size
$|\presence|+|\activity|$, which varies between $5$ and $10756$ and has an average of $1870$.  We
focused on a city map, because the density of the recorded points of
interest (POIs) in cities is significantly higher (and thus more
challenging) than in the countryside. We obtained the POIs from
OpenStreetMap\footnote{\url{http://www.openstreetmap.org}} (OSM) data.  In order to
assess on the usefulness of our approach we modeled the choice of
parameters as realistically as possible. However, the setting is
an example and can also be specified differently.

\subsection{Data and Experimental Setup}
The trajectories for our experiments were generated from random
shortest path queries on the OSM road network of Berlin. Each
trajectory is composed of a set of circular arcs and line segments as
described in Sect.~\ref{sec:transformation}. The viewport of the
camera is $800$ pixels wide and $600$ pixels high. Its speed and zoom
when moving along the trajectory is determined by the specified speed
limit of the underlying road. For each speed limit we introduce a zoom
level such that it takes at least $60$ seconds for a point to leave at
the bottom side of the viewport after entering the viewport on the top
side. This improves the legibility of labels moving through the
viewport. The change between two zoom levels is done by continuously
applying linear interpolation changing the zooming in reasonable time.
We took all POIs
which are tagged in OSM as \emph{fuel stations, parking lots, ATMs, restaurants, caf\'es, hotels, motels} and \emph{tourist information} as well as labels for \emph{countries, cities} and \emph{villages} --- a set we deemed suitable for car navigation systems. We further assigned a weight of 1 to every label and used the font
\emph{Helvetica} in point size 14 for rendering. It is enforced that any active
range of a label lasts at least one seconds to avoid flickering
labels. More sophisticated approaches comprising minimum visible area
of labels and minimum time between two active phases could be
incorporated easily. In this evaluation, however, we focus on the core
of our model.

All algorithms were implemented in C++ and compiled with GCC 4.8.3. ILPs were solved by Gurobi 6.0. All experiments were performed on an AMD Opteron 6172 processor clocked at 2.1 GHz, with 256 GB of RAM. Gurobi was allowed to use up to four cores in parallel, while all other experiments were run on a single core. Since we focus on \emph{Phase 2} in this
paper, we used an easy-to-implement approach for \emph{Phase 1} by
sampling the trajectory with high resolution. Much faster, but more
laborious approaches can be applied in practice as described in
Sect.~\ref{sec:transformation}. The evaluation of those approaches is
beyond the scope of this paper.

\subsection{Evaluation}

\begin{figure}[tb]
 \centering
  \includegraphics[page=1]{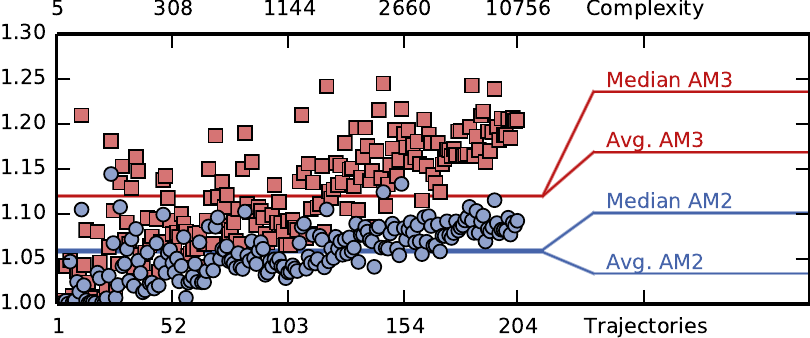}

 \caption{\small \GeneralMaxTotal, Comparison of Activity Models. Each data point represents an instance solved by
   \ILP. X-Axis:
   Instances are sorted by their complexity ($|\presence|+|\conflict|$) in increasing order. Y-Axis: Ratio between the optimal solution of AM2 (blue disks) or AM3 (red squares) and the optimal solution for AM1. }
 \label{plot:activity-models}
\end{figure}

For each trajectory we ran the different algorithms of Sect.~\ref{sec:phase2} for \GeneralMaxTotal and \kRestrictedMaxTotal (with $k=5$ and $k=10$) in the activity models AM1, AM2, and AM3. Any run exceeding
the time limit of $600$ seconds was aborted. Similarly when the graph
$G_{\text{AM}X}$ exceeded $10^7$ edges or vertices the run was  aborted to avoid memory overflow.
While for \IntervalSearch all runs were processed, the other
approaches did not complete all runs. Both \GreedySearch and
\LocalSearch completed about $92\%$ of the runs for \GeneralMaxTotal with
AM3; for the remaining $8\%$, the graph exceeded the aforementioned size limits. For all other model variants, all runs were completed. Further, \ILP sometimes exceeded its time limit for \kRestrictedMaxTotal with AM2 and AM3: for AM2, \ILP completed about $99\%$ and for AM3 about $66\%$ of the runs. For all other model variants, \ILP completed all its runs.

On each instance, we ran \LocalSearch $10$ times and report the average solution size. Each run was made with a different random seed and a time limit of $0.1$ seconds. We chose $0.1$ seconds, since we observed that \LocalSearch plateaus on nearly all instances after this time. Even with a $100$-fold increase to a time limit of $10$ seconds, we did not see significant improvement over the solution quality given after $0.1$ seconds (see Fig.~\ref{apx:plot:gmt:quality} in the appendix).

We now present extensive comparisons between the activity models,
optimization problems \GeneralMaxTotal and \kRestrictedMaxTotal, and
the applied algorithms.

\paragraph{Activity Models} We compare the activity models AM1, AM2,
and AM3 with each other by opposing the optimal solutions obtained by
\ILP. Figure~\ref{plot:activity-models} shows the ratio between the
solution for AM2 (AM3) and the solution for AM1 for \GeneralMaxTotal. By definition, a solution for AM1 is a lower bound for AM2, which again is a lower bound for AM3.
The activity is increased by a factor of 1.06 (1.12) on average for
AM2 (AM3).  Further, for \GeneralMaxTotal the ratio increases with increasing
complexity of the instances. Hence, for \GeneralMaxTotal the activity
models AM2 and AM3 increase the amount of displayed information
moderately. For general applications such as map
exploration this improvement is potentially helpful for the user.

In contrast, for 5-\textsc{RestrictedMaxTotal} the activity is
only increased by a factor of 1.02 (1.04) on average for AM2
(AM3);  see Fig.~\ref{apx:plot:activity-models} in the appendix. For 10-\textsc{Re\-strict\-ed\-Max\-Total} we obtain a factor of 1.03
(1.06) on average for AM2 (AM3). For both optimization problems this ratio decreases with the increasing complexity of the instances. Hence, for
\kRestrictedMaxTotal the activity models AM2 and AM3 increase the displayed amount of information only slightly, while producing more potentially
distracting visual effects by changing the labels' activities during their visibility in the viewport. Keeping in mind that \kRestrictedMaxTotal is
targeted for small screen devices such as smartphones and navigation systems, the measured gain of additional information does not necessarily justify the
additional visual distractions. Hence, AM2 and AM3 are less relevant in the context of \kRestrictedMaxTotal.

\begin{figure*}[tb]
\centering
\subfloat[Quality]{ \label{plot:gmt:quality}\includegraphics[page=1]{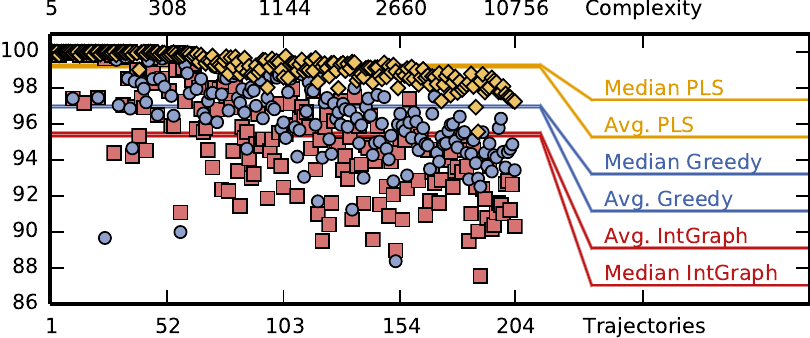}}
\hfill
\subfloat[Running Time]{ \label{plot:gmt:time}\includegraphics[page=1]{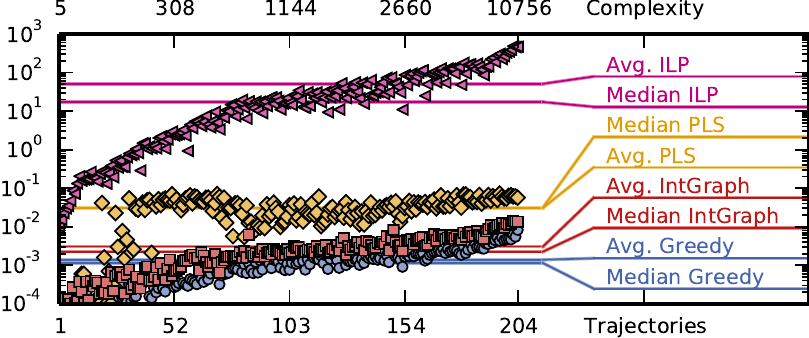}}
 \caption{\small \GeneralMaxTotal, AM1. Each data point represents an instance solved by
   \IntervalSearch (red square),  \GreedySearch (blue disk) or \LocalSearch (PLS) (yellow diamond). X-Axis:
   Instances are sorted by their complexity in increasing order. Y-Axis: (a) Achieved percentage of the optimal
   ILP solution. (b) Running time in seconds (logarithmic scale).}

\end{figure*}

\paragraph{Algorithms for \GeneralMaxTotal} 
Next, we compare the proposed algorithms with respect to \GeneralMaxTotal and AM1.
Figure~\ref{plot:gmt:quality} shows the activity obtained by single runs
in relation to the optimal solution obtained by \ILP.  In case
that \ILP exceeded the time limit, we used the upper bound that has
been found so far by \ILP as reference. If such an upper-bound has not
been found by \ILP, the run is omitted in the plot. Figure~\ref{plot:gmt:time} shows the running times, again with aborted runs omitted.

Concerning quality, \LocalSearch outperforms the
two other algorithms. No run achieved less than $95\%$ of the optimal
solution, while for \GreedySearch $23\%$ and for \IntervalSearch
$45\%$ of the runs achieved less than $95\%$ of the optimal
solution. On average \LocalSearch achieved $99\%$ of the optimal
solution, while \GreedySearch achieved $97\%$ and \IntervalSearch
achieved $96\%$ of the optimal solution.
Concerning average running times, \LocalSearch ($0.03$ sec.) is slower by about one order of magnitude compared to \GreedySearch ($0.001$ sec.) and \IntervalSearch ($0.003$ sec.). The running times of \ILP (average $51$ sec.) stayed far behind.

For AM2 and AM3 \LocalSearch is no longer the leader\footnote{The time for \LocalSearch to perform a single iteration depends upon the degree of vertices in the current independent set. Graphs for AM2 and AM3 have much higher vertex degrees than for AM1, which explains why \LocalSearch performs so poorly on these instances.} and \IntervalSearch
outperforms the other algorithms; see Fig.~\ref{apx:plot:gmt} in
the appendix. For AM2 both the average and median (about $89\%$) stay
behind the average and median of \GreedySearch (about $93\%$) and \IntervalSearch
(about $95\%$). For AM3 this gap is even more pronounced ($85\%$
vs. $90\%$ and $94\%$, respectively). Further, the quality of \LocalSearch is strongly
dispersed (minimum $67\%$).  For both activity models AM2 and AM3
\GreedySearch and \IntervalSearch yield similar results concerning
quality. However, concerning running time \IntervalSearch clearly
beats the other approaches and, unlike the other two
approaches, completed every run.

\begin{figure}[t]
\centering
\subfloat[\GeneralMaxTotal]{\label{fig:example-labeling:gmt} \includegraphics[page=1,width=0.48\linewidth]{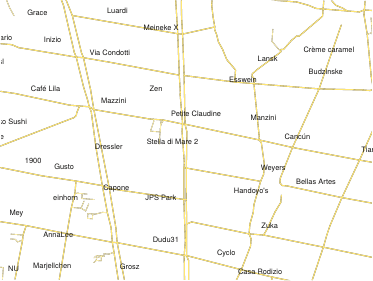}}
\hfill
\subfloat[10-\textsc{RestrictedMaxTotal}]{\label{fig:example-labeling:k10} \includegraphics[page=1,width=0.48\linewidth]{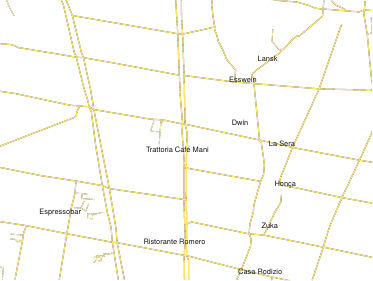}}
 \caption{\small Frame of a dynamic map labeling. While in (a) 54 labels are displayed at the same time, in (b) 10 labels are displayed in order to limit the informational content.}
 \label{fig:example-labeling}
\end{figure}

\paragraph{Optimization Models}
We now compare \GeneralMaxTotal with \kRestrictedMaxTotal.
For each trajectory and each integer $n < |L|$ we determined the proportion of the trajectory for which at least $n$ labels are active.
For \GeneralMaxTotal and AM1 we obtained the following results
(similar results hold for AM2 and AM3). On average for over $50\%$ of
the trajectory's length more than~$3$ labels are active at the same
time. However, for over~$25\%$ ($12.5\%$) of the trajectory's length
more than $8$ ($12$) labels are active at the same time, which already
may overwhelm untrained observers~\cite{Miller1956}. Further, for
$67\%$ ($42\%$) of the instances there are times when more than $20$
($40$) labels are active.  In some extreme cases over $60$ labels are
active at the same time. Figure~\ref{fig:example-labeling:gmt} shows a frame of a dynamic map labeling with $54$
active labels.  We
observe that in such cases the labels occupy a significant part of the
viewport and thus may occlude many other important map features.
For the
application of navigation systems and maps on smartphones it therefore lends itself to limit
the number of simultaneously active labels, which motivates
the relevance of \kRestrictedMaxTotal. Figure~\ref{fig:example-labeling:k10} shows the same
frame with a labeling produced by \ILP for $10$-\textsc{RestrictedMaxTotal}.

\paragraph{Algorithms for \kRestrictedMaxTotal}
Finally, we discuss the performance of the  algorithms for \kRestrictedMaxTotal with
$k=5$. Similar results hold for the case $k=10$; see
Fig.~\ref{apx:plot:k10} in the appendix.  Recall that \LocalSearch
does not support this optimization problem.  Figure~\ref{plot:k5} shows
the quality ratios and running times for
$5$-\textsc{Re\-stricted\-Max\-Total} in AM1; see Fig.~\ref{apx:plot:k5} in the appendix for AM2 and AM3.  \IntervalSearch
outperforms \GreedySearch both concerning quality and running time.
It achieves more than $99\%$ of the optimal solution on average. Further, every
run achieves at least $95\%$ of the optimal
solution. In contrast, \GreedySearch
achieves $96\%$ of the optimal solution on average. Further, $27\%$ of
the runs reach less than $95\%$ of the optimal solution, but at least
$89\%$.  While \IntervalSearch does not exceed a running time of $0.01$ seconds, \GreedySearch needs up to $0.1$ seconds. On average,
\IntervalSearch took $0.002$ seconds and \GreedySearch took $0.01$
seconds. Again, the running times of \ILP with an average of $54$ seconds stayed far behind.

\begin{figure*}[tb]
 \centering
  \subfloat[Quality]{\label{plot:k5:quality}\includegraphics[page=1]{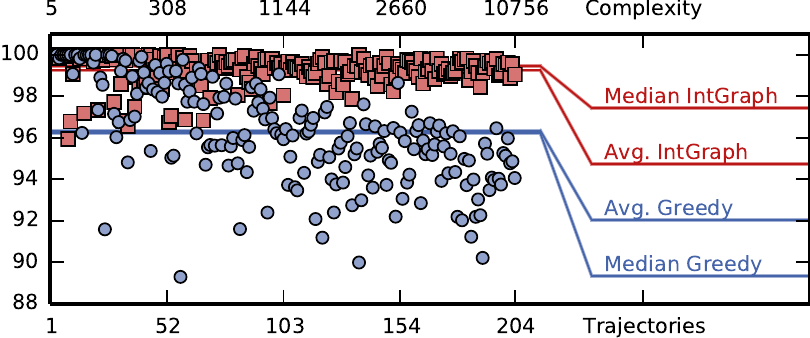}}
\hfill
  \subfloat[Running Time]{\label{plot:k5:time}\includegraphics[page=1]{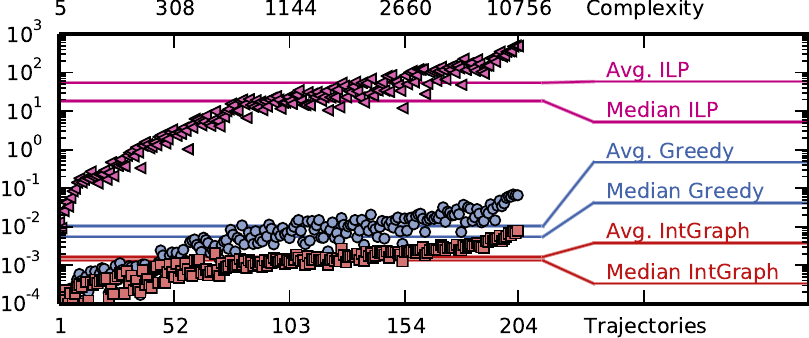}}
 \caption{\small $5$-\textsc{RestrictedMaxTotal}, AM1. Each data point represents an instance solved by
   \IntervalSearch (red square) or \GreedySearch (blue disk). X-Axis:
   Instances are sorted by their complexity ($|\presence|+|\conflict|$) in increasing order. Y-Axis:  (a) Achieved percentage of the optimal
   ILP solution. (b) Running time in seconds (logarithmic scale).}
 \label{plot:k5}
 \end{figure*}
\subsection{Discussion}

In our evaluation we considered both the temporal
labeling models and several labeling algorithms. From our comparison of the
three activity models we conclude that AM1, the most restricted model that does
not modify a label's activity during its presence interval and thus fully
avoids flickering, is not much worse in terms of the total activity. In fact,
the quality difference depends on the optimization problem: In
\GeneralMaxTotal the average improvement of AM2 is $6\%$ and of AM3 it is
$12\%$. For \kRestrictedMaxTotal and $k=10$ the average improvement of AM2 and AM3 is only
$3\%$ and $6\%$, respectively. Whether the gain in displayed content of AM2 and
AM3 outweighs the additional flickering effects would need to be examined in a
formal user study. Our evaluation of the models further shows that, without
placing any restrictions on the number of simultaneously active labels in
\GeneralMaxTotal, we frequently observe instances with relatively high numbers of labels,
which is not acceptable in certain applications---thus justifying the 
\kRestrictedMaxTotal model.

Our comparison of the algorithms showed that different algorithms are
preferable in different situations. For \GeneralMaxTotal and AM1
\LocalSearch outperformed all other algorithms in terms of
solution quality while being an order of magnitude slower than \GreedySearch and \IntervalSearch. However, for AM2 or AM3 or
\kRestrictedMaxTotal, \IntervalSearch is a clear winner
in both performance measures. \GreedySearch also
performs generally well and can be used as an easy-to-imple\-ment
approach. \ILP provides a simple way to compute optimal solutions and
was mainly used to evaluate the other algorithms in terms of solution
quality. It could be used directly as a solution approach, but its
running time is not reliable and external libraries are needed; thus,
we think that the other approaches are preferable in practice.

\section{Conclusion}
We presented a versatile and flexible temporal map labeling model that unifies and generalizes
several preceding models for dynamic map labeling.
Its strength lies in its purely combinatorial
nature, which abstracts away the problem's geometry. Thus, it can be used in any scenario where the start and end
times of label presences and conflicts can be determined in advance.
In a detailed experimental evaluation, we discussed the advantages of
different model variants and showed that simple and fast
algorithms yield near-optimal solutions for the application of navigation systems.

To apply our approach to maps exceeding the size of
city maps, we suggest decomposing the conflict graph into smaller
components. It seems likely that, when taking countrysides into
account, the conflict graph either already consists of several independent
components or it contains small cuts that allow for an appropriate
decomposition. Further, since our approach relies on algorithms for computing large weighted independent sets in graphs, this is another research
direction that promises improvements to our approach.

We focused on the core of our model in order to discuss its
application in general. However, with some engineering it can easily be extended to other scenarios or enhanced by further features such as a minimum visible area of labels, different types of map features or labels
avoiding obstacles.

\clearpage
\appendix

\section{Additional Diagrams}

\label{sec:additional-plots}
\centering

\begin{figure}[h]
  \centering
 \subfloat[AM1]{\includegraphics[page=1]{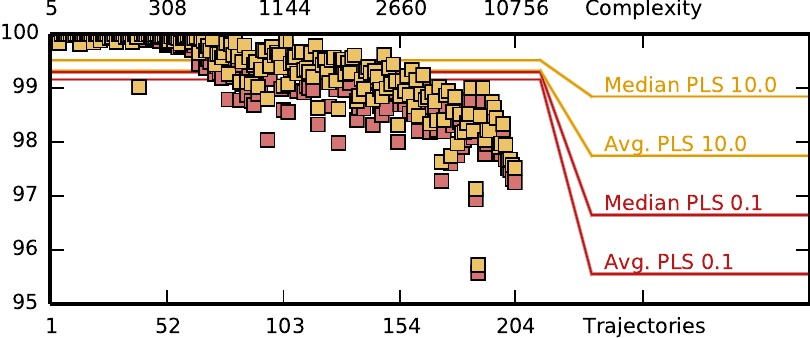}}

 \subfloat[AM2]{\includegraphics[page=1]{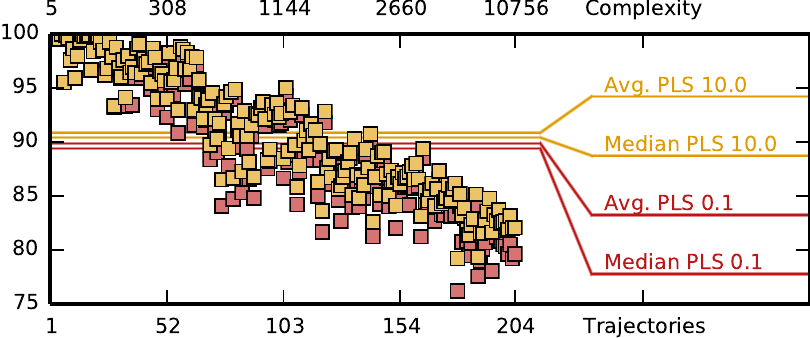}}

 \subfloat[AM3]{\includegraphics[page=1]{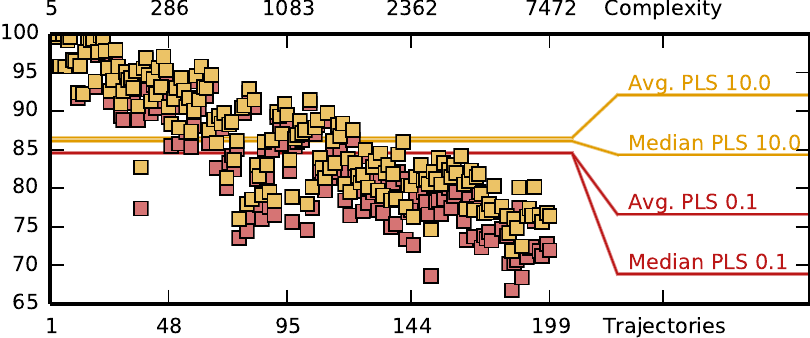}}
 \caption{\small \GeneralMaxTotal, Quality. Each data point represents an instance solved by
   \LocalSearch. The local search phase was aborted after 0.1 (red) and 10 (yellow) seconds. X-Axis:
   Instances are sorted by their complexity ($|\presence|+|\conflict|$) in increasing order. Y-Axis: Achieved percentage of the optimal
   ILP solution. }
 \label{apx:plot:gmt:quality}
\end{figure}

\begin{figure*}[tb]
  \centering
  \subfloat[\GeneralMaxTotal]{\includegraphics[page=1]{./plots/gmt_am.pdf}}

  \subfloat[10-\textsc{RestrictedMaxTotal}]{\includegraphics[page=1]{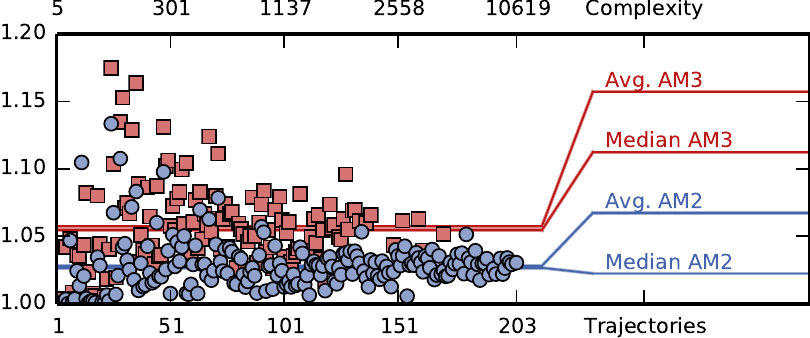}}

  \subfloat[5-\textsc{RestrictedMaxTotal}]{\includegraphics[page=1]{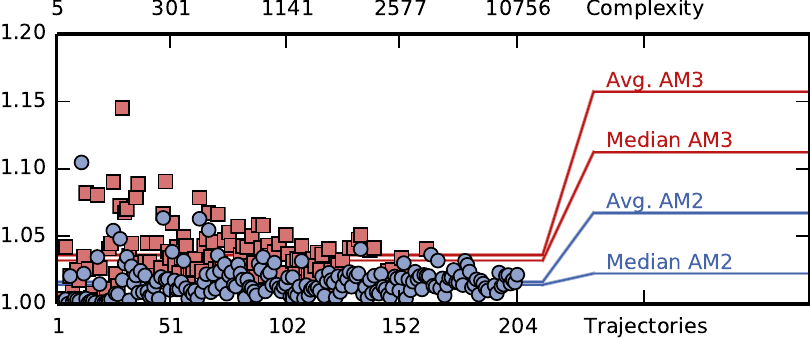}}

 \caption{\small Comparison of Activity Models. Each data point represents an instance solved by
   \ILP. X-Axis:
   Instances are sorted by their complexity ($|\presence|+|\conflict|$) in increasing order. Y-Axis: Ratio between the optimal solution of AM2 (AM3) and the optimal solution for AM1. }
 \label{apx:plot:activity-models}
\end{figure*}

\begin{landscape}
\begin{figure*}[tb]
 \centering
 \subfloat[Quality, AM1]{\includegraphics[page=1]{./plots/gmt_quality_am1.pdf}}\hspace{1cm}
 \subfloat[Running Time, AM1]{\includegraphics[page=1]{./plots/gmt_time_am1.pdf}}

 \subfloat[Quality, AM2]{\includegraphics[page=1]{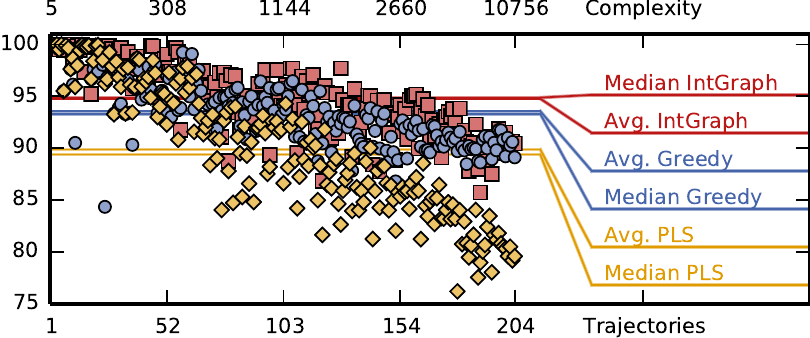}}\hspace{1cm}
 \subfloat[Running Time, AM2]{\includegraphics[page=1]{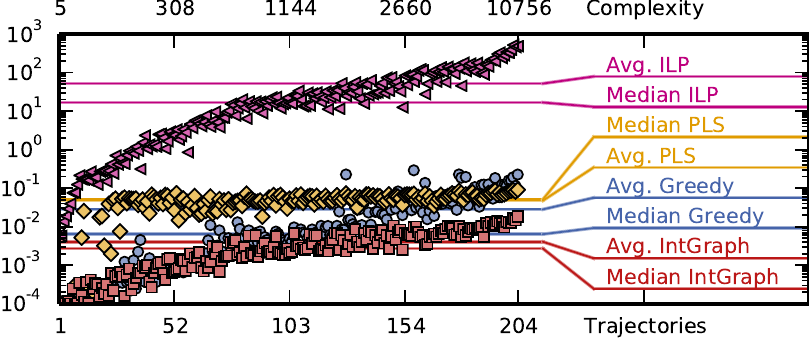}}

 \subfloat[Quality, AM3]{\includegraphics[page=1]{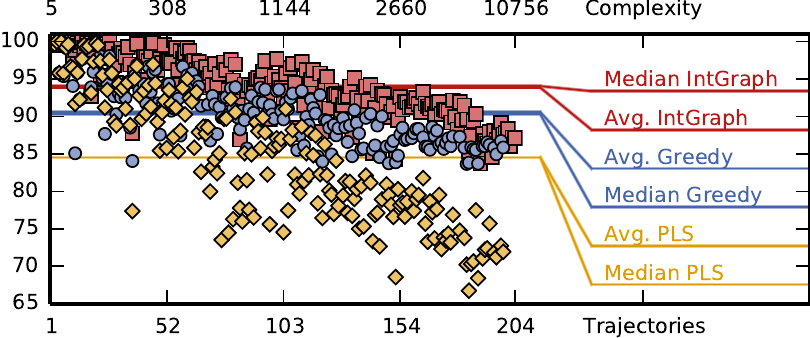}}\hspace{1cm}
 \subfloat[Running Time, AM3]{\includegraphics[page=1]{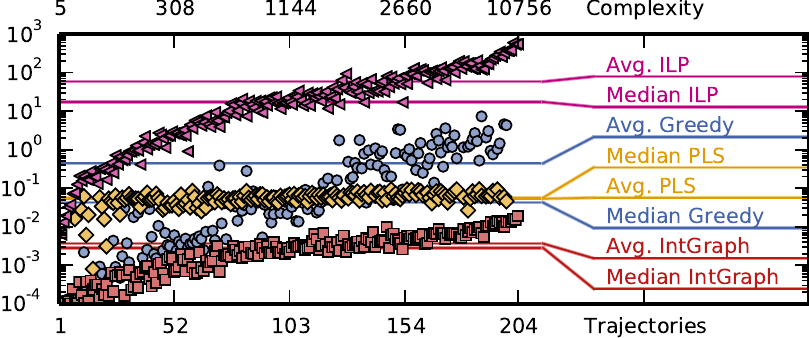}}

 \caption{\small\GeneralMaxTotal. Each data point represents an instance solved by
   \IntervalSearch (red square),  \GreedySearch (blue disk) or \LocalSearch (yellow diamond). X-Axis:
   Instances are sorted by their complexity ($|\presence|+|\conflict|$) in increasing order. Y-Axis: (a),(c),(e): Achieved percentage of the optimal
   ILP solution. (b),(d),(f): Running time in seconds.}
 \label{apx:plot:gmt}
\end{figure*}
\end{landscape}

\begin{landscape}
  \begin{figure*}[tb]
    \centering
    \subfloat[Quality,
    AM1]{\includegraphics[page=1]{./plots/k5_quality_am1.pdf}}\hspace{1cm}
    \subfloat[Running Time,
    AM1]{\includegraphics[page=1]{./plots/k5_time_am1.pdf}}

    \subfloat[Quality,
    AM2]{\includegraphics[page=1]{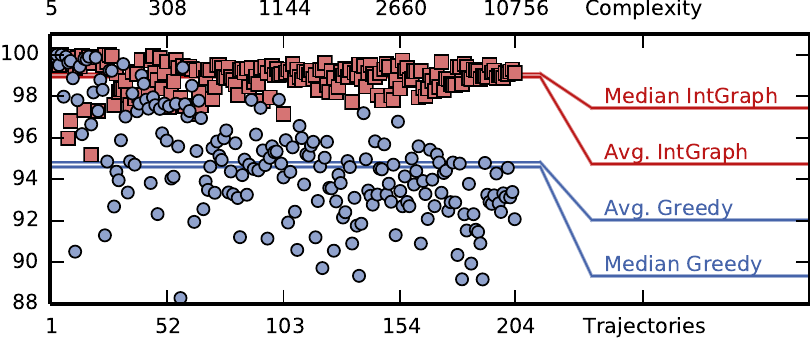}}\hspace{1cm}
    \subfloat[Running Time,
    AM2]{\includegraphics[page=1]{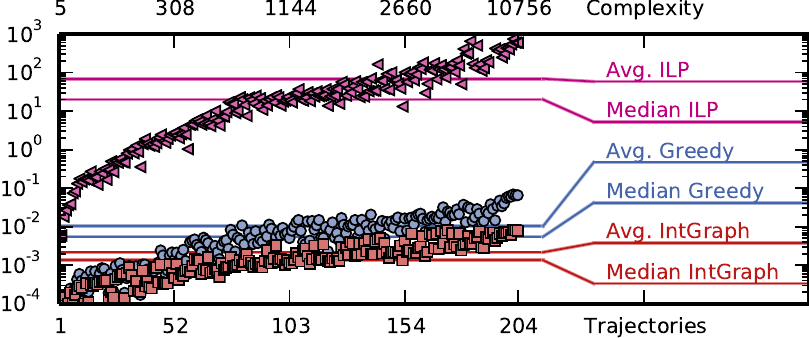}}

    \subfloat[Quality,
    AM3]{\includegraphics[page=1]{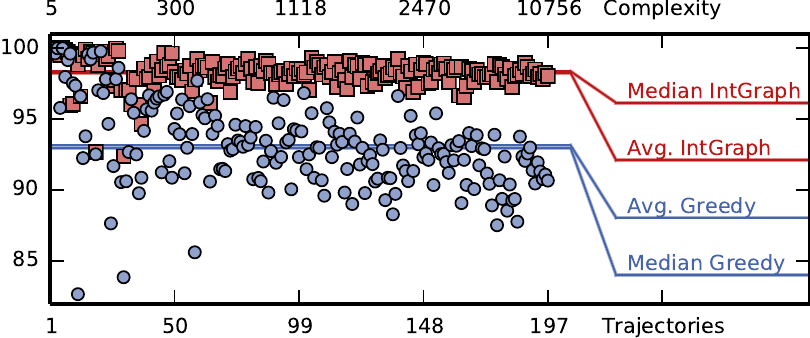}}\hspace{1cm}
    \subfloat[Running Time,
    AM3]{\includegraphics[page=1]{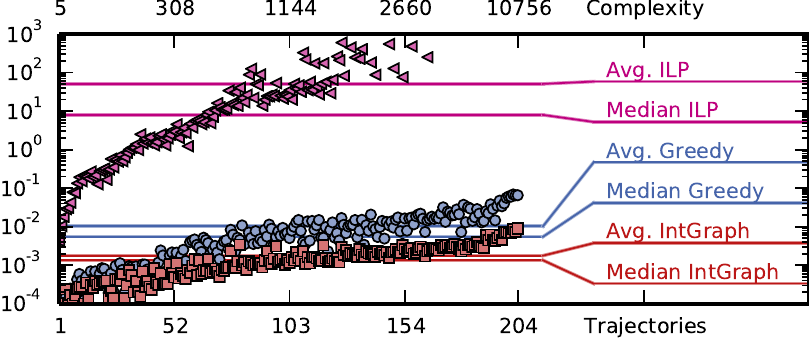}}

    \caption{\small $5$-\textsc{RestrictedMaxTotal}, Quality. Each
      data point represents an instance solved by \IntervalSearch (red
      square) or \GreedySearch (blue disk). X-Axis: Instances are
      sorted by their complexity ($|\presence|+|\conflict|$) in
      increasing order. Y-Axis: (a),(c),(e): Achieved percentage of
      the optimal ILP solution. (b),(d),(f): Running time in seconds.}
    \label{apx:plot:k5}
  \end{figure*}
\end{landscape}

\begin{landscape}
\begin{figure*}[tb]
 \centering 
 \subfloat[AM1]{\includegraphics[page=1]{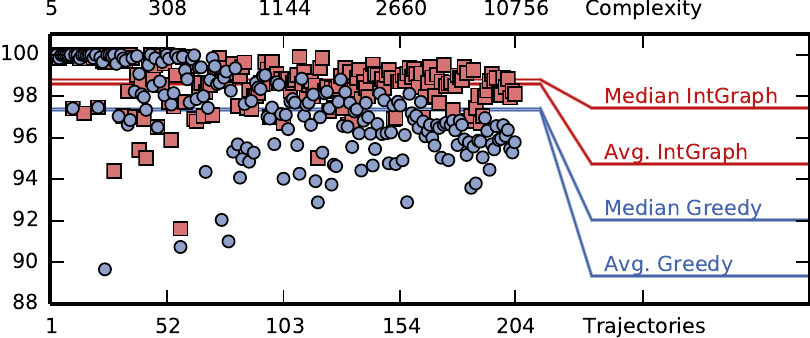}}\hspace{1cm}
 \subfloat[AM1]{\includegraphics[page=1]{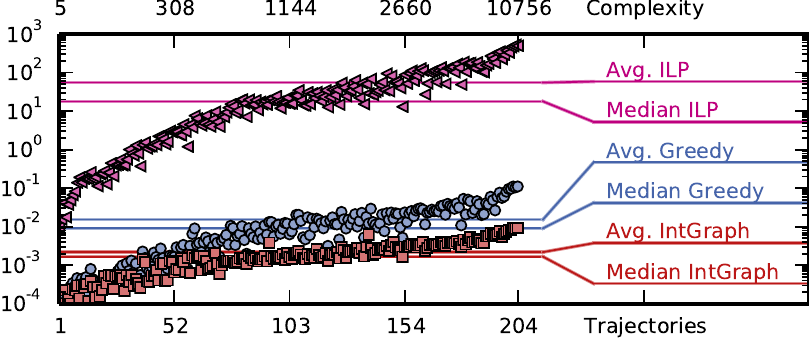}}

 \subfloat[AM2]{\includegraphics[page=1]{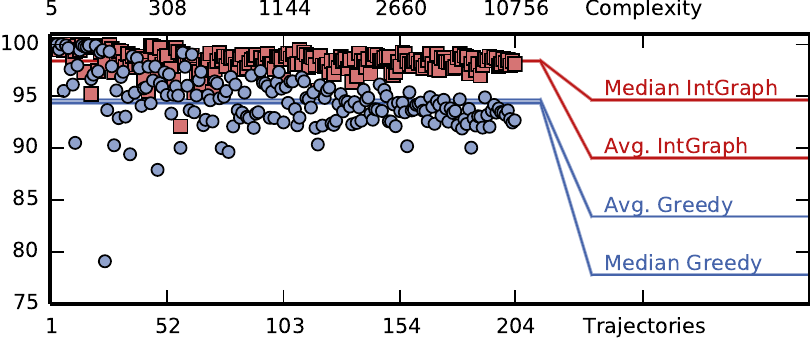}}\hspace{1cm}
  \subfloat[AM2]{\includegraphics[page=1]{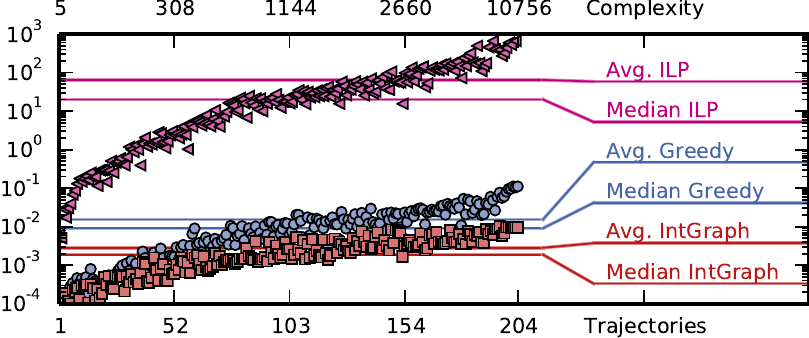}}

 \subfloat[AM3]{\includegraphics[page=1]{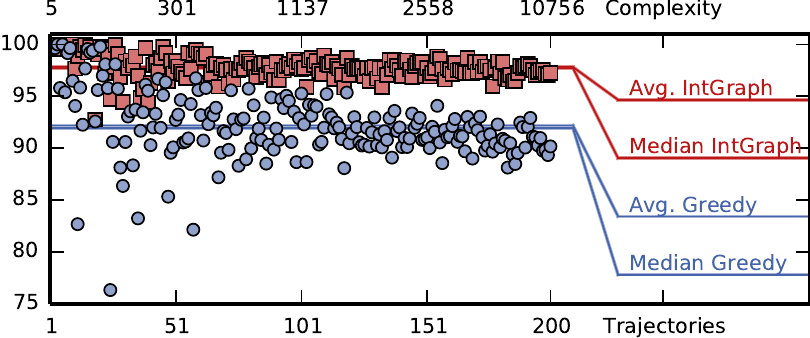}}\hspace{1cm}
 \subfloat[AM3]{\includegraphics[page=1]{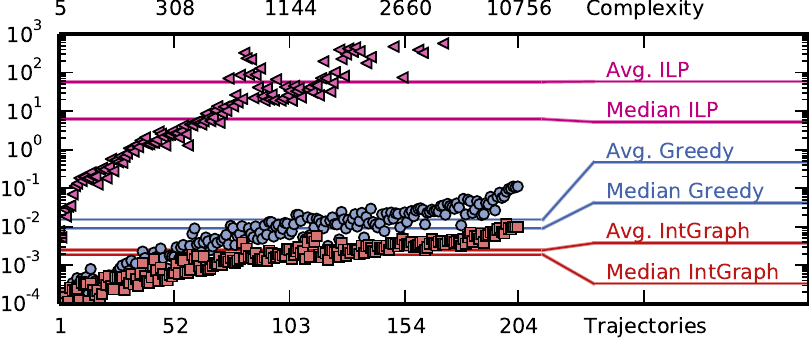}}

 \caption{\small$10$-\textsc{RestrictedMaxTotal}, Quality. Each data point represents an instance solved by
   \IntervalSearch (red square) or \GreedySearch (blue disk). X-Axis:
   Instances are sorted by their complexity ($|\presence|+|\conflict|$) in increasing order. Y-Axis: (a),(c),(e): Achieved percentage of the optimal
   ILP solution. (b),(d),(f): Running time in seconds.}
 \label{apx:plot:k10}
\end{figure*}
\end{landscape}


\begin{thebibliography}{10}

\bibitem{Been2006}
K.~Been, E.~Daiches, and C.~Yap.
\newblock Dynamic map labeling.
\newblock {\em IEEE Trans. Visual. Comput. Graphics}, 12(5):773--780, 2006.

\bibitem{Been2010}
K.~Been, M.~Nöllenburg, S.-H. Poon, and A.~Wolff.
\newblock Optimizing active ranges for consistent dynamic map labeling.
\newblock {\em Comput. Geom. Theory Appl.}, 43(3):312--328, 2010.

\bibitem{Buchin2014}
K.~Buchin and D.~Gerrits.
\newblock Dynamic point labeling is strongly {PSPACE}-complete.
\newblock {\em Int. J. Comput. Geom. Appl.}, 24(4):373--395, 2014.

\bibitem{DeBerg2013}
M.~de~Berg and D.~H.~P. Gerrits.
\newblock Labeling moving points with a trade-off between label speed and label
  overlap.
\newblock In {\em Algorithms (ESA'13)}, volume 8125 of {\em LNCS}, pages
  373--384. Springer, 2013.

\bibitem{fw-ppwalm-91}
M.~Formann and F.~Wagner.
\newblock A packing problem with applications to lettering of maps.
\newblock In {\em Computational Geometry (SoCG'91)}, pages 281--288. ACM, 1991.

\bibitem{Gemsa2013}
A.~Gemsa, B.~Niedermann, and M.~Nöllenburg.
\newblock Trajectory-based dynamic map labeling.
\newblock In {\em Algorithms and
  Computation (ISAAC'13)}, volume 8283 of {\em LNCS}, pages 413--423. Springer, 2013.

\bibitem{Gemsa2011}
A.~Gemsa, M.~N{\"{o}}llenburg, and I.~Rutter.
\newblock Sliding labels for dynamic point labeling.
\newblock In {\em Canadian Conf.  Comput. Geom. (CCCG'11)},
  pages 205--210, 2011.

\bibitem{gnr-clrm-16}
A.~Gemsa, M.~Nöllenburg, and I.~Rutter.
\newblock Consistent labeling of rotating maps.
\newblock {\em J. Computational Geometry}, 7(1):308--331, 2016.

\bibitem{gnr-elsrm-16}
A.~Gemsa, M.~Nöllenburg, and I.~Rutter.
\newblock Evaluation of labeling strategies for rotating maps.
\newblock {\em {ACM} J. Experimental Algorithmics}, 21(1):1.4:1--1.4:21, 2016.

\bibitem{Hsiao1992}
J.~Y. Hsiao, C.~Y. Tang, and R.~S. Chang.
\newblock An efficient algorithm for finding a maximum weight 2-independent set
  on interval graphs.
\newblock {\em Inform. Process. Lett.}, 43(5):229--235, 1992.

\bibitem{i-pnm-75}
E.~Imhof.
\newblock Positioning names on maps.
\newblock {\em The American Cartographer}, 2(2):128--144, 1975.

\bibitem{Jin201520}
Y.~Jin and J.-K. Hao.
\newblock General swap-based multiple neighborhood tabu search for the maximum
  independent set problem.
\newblock {\em Eng. Appl. Artif. Intel.}, 37:20--33, 2015.

\bibitem{kb-alppcilp-08}
J.~Kern and C.~Brewer.
\newblock Automation and the map label placement problem: A comparison of two
  {GIS} implementations of label placement.
\newblock {\em Cartographic Perspectives}, 60:22--45, 2008.

\bibitem{Liao2014}
C.-S. Liao, C.-W. Liang, and S.-H. Poon.
\newblock Approximation algorithms on consistent dynamic map labeling.
\newblock In {\em Frontiers in Algorithmics (FAW'14)}, volume 8497 of {\em
  LNCS}, pages 170--181. Springer, 2014.

\bibitem{Luboschik2008}
M.~Luboschik, H.~Schumann, and H.~Cords.
\newblock Particle-based labeling: Fast point-feature labeling without
  obscuring other visual features.
\newblock {\em IEEE Trans. Visual. Comput. Graphics}, 14(6):1237--1244, 2008.

\bibitem{Maass2006}
S.~Maass and J.~D{\"{o}}llner.
\newblock Efficient view management for dynamic annotation placement in virtual
  landscapes.
\newblock In {\em Smart Graphics (SG'06)}, volume 4073 of {\em LNCS}, pages
  1--12. Springer, 2006.

\bibitem{Miller1956}
G.~A. Miller.
\newblock The magical number seven, plus or minus two: Some limits on our
  capacity for processing information.
\newblock {\em Psychological Review}, 63(2):81--97, 1956.

\bibitem{Mote2007}
K.~Mote.
\newblock Fast point-feature label placement for dynamic visualizations.
\newblock {\em Information Visualization}, 6(4):249--260, 2007.

\bibitem{n-cldmust-12}
B.~Niedermann.
\newblock Consistent labeling of dynamic maps using smooth trajectories.
\newblock Master's thesis, Karlsruhe Institute of Technology, June 2012.

\bibitem{Petzold2003}
I.~Petzold, G.~Gr{\"{o}}ger, and L.~Pl{\"{u}}mer.
\newblock {Fast screen map labeling - Data structures and algorithms}.
\newblock In {\em International Cartographic Conference (ICC'03)}, pages
  288--298, 2003.

\bibitem{Pullan2006}
W.~Pullan.
\newblock Phased local search for the maximum clique problem.
\newblock {\em J. Comb. Optim.}, 12(3):303--323, 2006.

\bibitem{Pullan2009214}
W.~Pullan.
\newblock Optimisation of unweighted/weighted maximum independent sets and
  minimum vertex covers.
\newblock {\em Discrete Optim.}, 6(2):214--219, 2009.

\bibitem{Yokosuka2013}
Y.~Yokosuka and K.~Imai.
\newblock Polynomial time algorithms for label size maximization on rotating
  maps.
\newblock In {\em Canadian Conf.  Comput. Geom. (CCCG'13)},
  pages 187--192, 2013.

\bibitem{Zhang2015}
X.~Zhang, S.-H. Poon, M.~Li, and V.~Lee.
\newblock On maxmin active range problem for weighted consistent dynamic map
  labeling.
\newblock In {\em GEOProcessing 2015}, pages 32--37. IARIA, 2015.

\end{thebibliography}
\end{document}